\documentclass[conference]{IEEEtran}
\IEEEoverridecommandlockouts
\usepackage{cite}
\usepackage{amsmath,amssymb,amsfonts}
\usepackage{algorithm}
\usepackage{algorithmic}
\usepackage{graphicx}
\usepackage{textcomp}
\usepackage{xcolor}
\usepackage{amsthm}
\usepackage{caption}
\usepackage{setspace}
\usepackage{balance}
\usepackage[
top    = 0.73 in,
bottom = 1.03 in,
left   = 0.57 in,
right  = 0.572 in]{geometry}   

\makeatletter
\renewcommand{\maketag@@@}[1]{\hbox{\m@th\normalsize\normalfont#1}}%
\makeatother
\allowdisplaybreaks[4]
\newtheorem{lemma}{Lemma}

\def\BibTeX{{\rm B\kern-.05em{\sc i\kern-.025em b}\kern-.08em
    T\kern-.1667em\lower.7ex\hbox{E}\kern-.125emX}}
\begin{document}

\title{Multiuser Resource Allocation for Semantic-Relay-Aided Text Transmissions\\
}

\author{\IEEEauthorblockN{Zeyang Hu\IEEEauthorrefmark{1}, Tianyu Liu\IEEEauthorrefmark{1}, Changsheng You\IEEEauthorrefmark{1}, Zhaohui Yang\IEEEauthorrefmark{2}, and Mingzhe Chen\IEEEauthorrefmark{3}} 
\IEEEauthorblockA{\IEEEauthorrefmark{1}\text{Department of Electronic and Electrical Engineering}, \text{Southern University of Science and Technology}, Shenzhen, China\\
}
\IEEEauthorblockA{\IEEEauthorrefmark{2}\text{College of Information Science and Electronic Engineering}, \text{Zhejiang University}, Hangzhou, China\\
}
\IEEEauthorblockA{\IEEEauthorrefmark{3}\text{Department of Electrical and Computer Engineering and Institute for Data Science and Computing}, \\
\text{University of Miami}, Coral Gables, FL, 33146, USA\\ 
Emails: \{huzy2022, liuty2022\}@mail.sustech.edu.cn, youcs@sustech.edu.cn, \\zhaohui\_yang@zju.edu.cn, mingzhe.chen@miami.edu.
}
}

\maketitle
\begin{abstract}
Semantic communication (SemCom) is an emerging technology that extracts useful meaning from data and sends only relevant semantic information. 
Thus, it has the great potential to improve the spectrum efficiency of conventional wireless systems with bit transmissions, especially in low signal-to-noise ratio (SNR) and small bandwidth regions. 
However, the existing works have mostly overlooked the constraints of mobile devices, which may not have sufficient capabilities to implement resource-demanding semantic encoder/decoder based on deep learning. 
To address this issue, we propose in this paper a new \emph{semantic relay} (SemRelay), which is equipped with a semantic receiver to assist multiuser text transmissions. 
Specifically, the SemRelay decodes semantic information from a base station and forwards it to the users using conventional bit transmission, hence effectively improving text transmission efficiency. 
To study the multiuser resource allocation, we formulate an optimization problem to maximize the multiuser weighted sum-rate by jointly designing the SemRelay transmit power allocation and system bandwidth allocation. 
Although this problem is non-convex and hence challenging to solve, we propose an efficient algorithm to obtain its high-quality suboptimal solution by using the block coordinate descent method. 
Last, numerical results show the effectiveness of the proposed algorithm as well as superior performance of the proposed SemRelay over the conventional decode-and-forward (DF) relay, especially in small bandwidth region. 

\end{abstract}

\begin{IEEEkeywords}
Semantic communication, semantic relay, joint power and bandwidth allocation.
\end{IEEEkeywords}

\section{Introduction}
Semantic communication (SemCom) has recently emerged as a promising technology to enhance the data transmission efficiency by leveraging advanced artificial intelligence (AI) techniques  \cite{yang2022semantic, SC_magazine_beyond, lan2021semantic}. 
Thus, it has the great potential to break the Shannon's capacity limits in future sixth-generation (6G) wireless systems. 
Specifically, instead of transmitting bit sequences in conventional wireless systems, SemCom conveys the semantic meaning embedded in the source data to reduce the communication overhead and improve resource utilization efficiency at the semantic level \cite{SC_magazine_beyond}. 
These advantages thus have motivated upsurging research interests in recent years to apply SemCom technologies in various applications, such as education and medical treatment through virtual reality, human-machine interaction with multiple intelligent devices, and collaborative communication between vehicles, which impose stringent end-to-end latency and data rate requirements. 

The existing works on SemCom have mostly focused on three key aspects, namely, SemCom system architectures, semantic representation (or recovery), and resource management. 
Specifically, the recent advancement in deep learning (DL) techniques and the improvement in device computation and storage capabilities have enabled the widespread utilization of deep neural networks in designing SemCom systems. 
For example, a new DL-enabled SemCom system (DeepSC) was proposed in \cite{xie2021deep}, which leverages transformer based techniques for designing joint semantic-and-channel coding. 
Based on DeepSC frameworks, the authors in \cite{xie2022task} and \cite{zhang2022unified} further proposed a new task-oriented multi-modal data SemCom system (MU-DeepSC) and a unified deep learning enabled SemCom (U-DeepSC) system, respectively, for multi-task scenarios, which enable efficient multi-modal data transmissions. 
In \cite{zhou2022cognitive}, the authors proposed a cognitive SemCom framework that employs a knowledge graph shared between the transmitter and receiver, which was demonstrated to outperform other benchmark systems in terms of data compression rate and communication reliability. 
Besides, a graph convolutional networks based (GCN-based) semantic model was proposed in \cite{xiao2022imitation} to construct a better knowledge base available to all edge servers in SemCom, which significantly reduces the required computation, storage, and communication resources. 
Next, in order to exploit the useful meaning of semantic information, a novel multilayer semantic representation approach was proposed in \cite{xiao2022imitation}, which incorporated explicit and implicit semantics. 
In addition, knowledge graph (KG) is also a promising approach to represent and recover semantic information especially in text semantic transmission \cite{wang2022performance}. Third, for SemCom resource management, several representative performance metrics have been proposed recently, such as the semantic rate (S-R), semantic spectral efficiency (S-SE), and semantic similarity \cite{yan2022resource, mu2022semi_NOMA,li2023NOMA_multi_user}. 
For example, the semantic similarity was defined in \cite{yan2022resource} to evaluate the accuracy of semantic transmission and determine whether semantic transmission can satisfy the transmission requirements. 
Besides, a novel rate-splitting method was proposed in \cite{yang2023energy} to reduce communication and computation energy consumption.  
The aforementioned studies have demonstrated that SemCom has the great potential to achieve superior transmission efficiency over conventional bit transmission, especially in low signal-to-noise ratio (SNR) and small-bandwidth scenarios. 
{{However, most prior works in SemCom have simply assumed the deployment of DL neural networks (e.g., DeepSC receivers), on mobile devices.}} 
This assumption, however, overlooks the fact that normal mobile devices may have limited computation and storage resources and thus may not be able to implement DL-based SemCom techniques. 

To overcome the above limitations, we propose in this paper a novel semantic relay (SemRelay) architecture as shown in Fig. \ref{fig_structure}, which is properly deployed to assist multiuser text transmissions from a base station (BS) with a DeepSC transmitter to multiple resource-constrained users. 
Note that the proposed SemRelay differs from the conventional decode-and-forward (DF) relay designed for bit transmission. 
Specifically, equipped with a well-trained DeepSC receiver, the SemRelay decodes the semantic information sent by the BS and forwards it to the users using conventional bit-based transmission. 
This approach thus not only improves the text transmission efficiency through SemCom over the BS$\rightarrow$SemRelay link, but also releases the computation and storage burden at the resource-constrained users, since the computation-demanding semantic decoding is performed at the SemRelay. 
For the considered SemRelay-aided multiuser text transmission system, we assume that the BS$\rightarrow$SemRelay semantic transmission and SemRelay$\rightarrow$user bit transmissions are operated in orthogonal frequency bands. 
An optimization problem is formulated to maximize the achievable weighted sum-(bit)-rate by joint designing the SemRelay transmit power allocation and system bandwidth allocation. 
However, this problem is non-convex due to the intricately coupled optimization variables, thus rendering it difficult to obtain the optimal solution. 
To address this issue, we propose a block coordinate descent (BCD) based algorithm to obtain its high-quality suboptimal solution by alternately optimizing the SemRelay transmit power allocation and system bandwidth allocation, with the other one being fixed. 
Our numerical results demonstrate the effectiveness of the proposed algorithm and the significant rate performance gain achieved by the proposed SemRelay as compared to the conventional DF relay. 

\begin{figure}[!t]
    \vspace{0.07 in}
    \centerline
    {\includegraphics[width=0.85\linewidth]{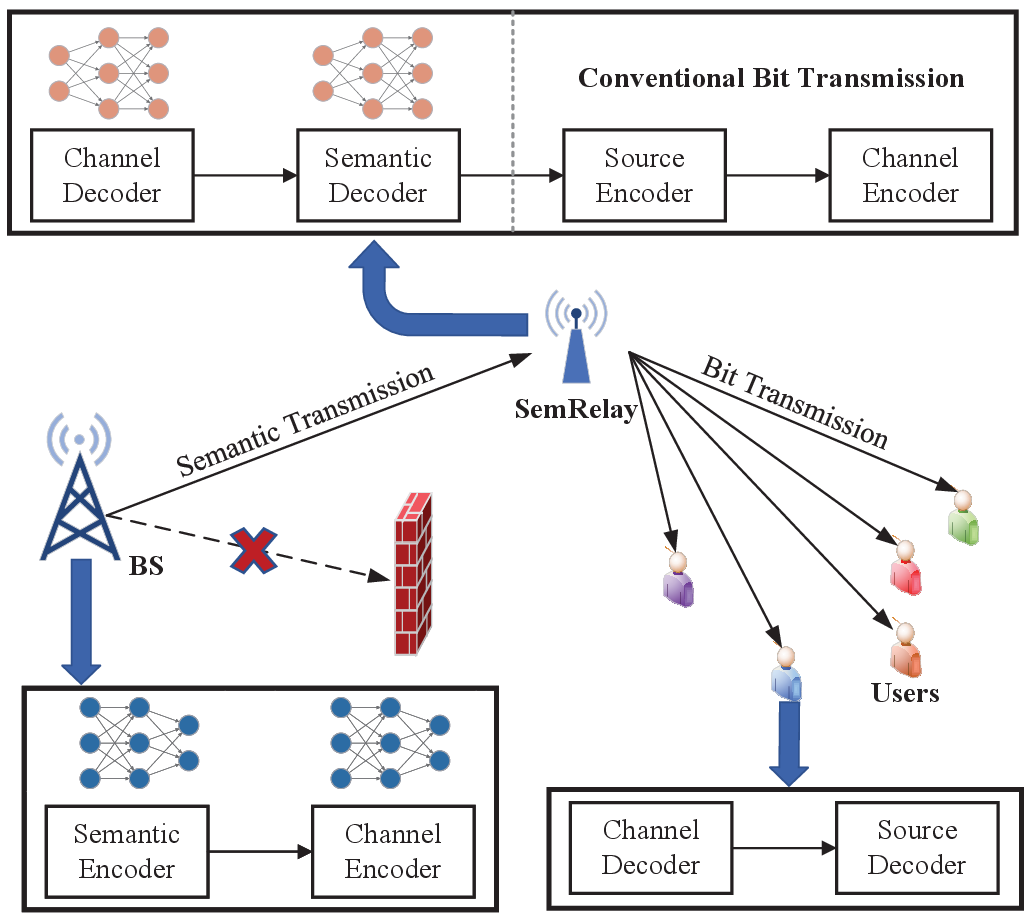}}
    \caption{A SemRelay-aided multiuser text transmission system.}
    \label{fig_structure}
\end{figure}

\section{System Model and Problem Formulation}

\subsection{System Model}

We consider a SemRelay-aided multiuser text transmission system as shown in Fig. \ref{fig_structure}, where a BS transmits text information to $N$ users, all equipped with a single antenna\footnote{This work can be easily extended to the case of multi-antenna BS and users.}. 
We assume that the BS has abundant storage and computation resources for enabling SemCom via deploying a well-trained DeepSC encoder, while the users have limited storage and computation capabilities and thus are unable to decode the semantic information transmitted by the BS. 
Moreover, the direct links between the BS and users are assumed to be blocked due to long distances and obstacles. 

To enable efficient text transmission from the BS to the users, a SemRelay with a DeepSC decoder and abundant computation resources is properly deployed near the users to assist the multiuser text transmissions in two phases. 
Specifically, during the first phase, the BS sends the text information for all the users to the SemRelay via its DeepSC transmitter. 
In the second phase, the SemRelay first decodes the text information from received signals using its channel-and-semantic decoders, and then forwards them to the corresponding users via conventional bit-based transmission. 
We assume that the semantic decoding delay at the SemRelay is negligible thanks to its abundant computation resources. 
Moreover, to avoid communication interference, we adopt the frequency-division multiple access (FDMA) scheme for the SemRelay-aided multiuser text transmissions, where the semantic transmission over the BS$\rightarrow$SemRelay link and bit transmissions over the SemRelay$\rightarrow$user links are operated in separate frequency bands. 


\subsubsection{BS-SemRelay semantic transmission}
Let $\boldsymbol{s}$ denote the original text that the BS transmits to the SemRelay. 
In order to extract the semantic information from $\boldsymbol{s}$ and transmit it accurately over wireless fading channels, $\boldsymbol{s}$ is encoded through the semantic-and-channel encoder. 
In particular, the encoded signal is represented as 
{\setlength\abovedisplayskip{1pt}
\setlength\belowdisplayskip{2pt}
\begin{equation}
    \boldsymbol{x_{\mathrm{br}}} = \mathcal{C}_{s}(\boldsymbol{s}) \in \mathbb{C}^{1 \times M}, 
\end{equation}}%
where $\mathcal{C}_{s}(\cdot)$ denotes the semantic-channel encoder, $M$ is the average number of mapped semantic symbols for $\boldsymbol{s}$. 
Moreover, let $h_{\mathrm{br}}$ denote the channel from the BS to the SemRelay, 
$B$ denote the total bandwidth of the system, and $\alpha_{\mathrm{br}}$ denote the fraction of bandwidth allocated to the BS$\rightarrow$SemRelay link. 
Then the received SNR in dB at the SemRelay, denoted by $\gamma_{\mathrm{br}}$, is given by
\begingroup\makeatletter\def\f@size{8}\check@mathfonts
{\setlength\abovedisplayskip{1pt}
\setlength\belowdisplayskip{2pt}
\begin{equation}
        \gamma_{\mathrm{br}} = 10\log_{10}{\left(\frac{\left|h_{\mathrm{br}}\right|^{2}P_{\mathrm{b}}}{\alpha_{\mathrm{br}}BN_{0}}\right)},
\end{equation}}%
\endgroup
where $P_{\mathrm{b}}$ is the transmit power of the BS,  $N_{0}$ denotes the power spectral density of the AWGN. 
Different from conventional bit transmission, semantic transmission utilizes the semantic symbols to convey semantic information. 
Let $L$ denote the average number of words per sentence, $K$ denote the average number of semantic symbols to represent each word, and $I$ represent the average semantic information per sentence measured in semantic units (suts). 
Hence, the effective semantic rate (in suts/s) in the BS$\rightarrow$SemRelay link, which represents the average amount of successfully transmitted semantic information per second, can be characterized as \cite{mu2022semi_NOMA,li2023NOMA_multi_user}
{\setlength\abovedisplayskip{1pt}
\setlength\belowdisplayskip{2pt}
\begin{equation}
        R_{\mathrm{S}} = \frac{\alpha_{\mathrm{br}}BI}{KL} \varepsilon_{K}(\gamma_{\mathrm{br}}). \label{R_s}
\end{equation}}%
\noindent Herein, $\alpha_{\mathrm{br}}B$ is the communication bandwidth of the BS$\rightarrow$SemRelay link, which equals the value of the symbol rate. 
$\varepsilon_{K}(\gamma_{\mathrm{br}})\in[0,1]$ is the so-called semantic similarity, which measures the distance of semantic information between the original and recovered texts \cite{yan2022resource}. 
As shown in \cite{mu2022semi_NOMA} and \cite{li2023NOMA_multi_user}, the semantic similarity is determined by $K$ and $\gamma_{\mathrm{br}}$. 
Specifically, $\varepsilon_{K}(\gamma_{\mathrm{br}})$ generally follows an `S'-shaped curve with respect to $\gamma_{\mathrm{br}}$ for any $K$ \cite{mu2022semi_NOMA}. 
As such, $\varepsilon_{K}(\gamma_{\mathrm{br}})$ can be approximated by using a generalized logistic function, which is given by 
{\setlength\abovedisplayskip{1pt}
\setlength\belowdisplayskip{2pt}
\begin{equation}
    \varepsilon_{K}(\gamma_{\mathrm{br}}) = a_1 + \frac{a_2}{1+e^{-(c_{1}\gamma_{\mathrm{br}}+c_{2})}}, \label{sesi}
\end{equation}}%
where $a_1, a_2, c_1$, and $c_2$ are constant parameters determined by $K$, which can be numerically obtained. 
To ensure effective text transmission, the semantic similarity should be no smaller than a threshold, i.e., $\varepsilon_{K}(\gamma_{\mathrm{br}}) \geq \bar{\varepsilon}$, where $\bar{\varepsilon}$ is the required minimum semantic similarity \cite{yan2022resource}. 
Moreover, we denote the average number of bits contained in each word by $\mu$ (in bits/word) in conventional text transmission \cite{mu2022semi_NOMA,yan2022resource}. 
Then the (effective) semantic-to-bit rate in the BS$\rightarrow$SemRelay link can be obtained as 
{\setlength\abovedisplayskip{1pt}
\setlength\belowdisplayskip{2pt}
\begin{equation}
        R_{\mathrm{br}} = \mu\frac{R_{\mathrm{S}}}{I/L} =\mu\frac{\alpha_{\mathrm{br}}B}{K} \varepsilon_{K}(\gamma_{\mathrm{br}}). \label{R_br}
\end{equation}}

\addtolength{\topmargin}{-0.03in}
\subsubsection{SemRelay-user bit transmission}
With semantic signals received at the SemRelay, the SemRelay first decodes the text information by the DeepSC receiver and then forwards it to the users via conventional bit transmission in orthogonal bandwidth. 
For each user $n \in \mathcal{N} = \{1,\cdots,N\}$, let $h_{\mathrm{ru}}^{(n)}$ denote its channel from the SemRelay. 
Then its achievable bit rate in bps is
\begingroup\makeatletter\def\f@size{8}\check@mathfonts
{\setlength\abovedisplayskip{1pt}
\setlength\belowdisplayskip{2pt}
\begin{equation}
    \begin{small}
        R_{\mathrm{ru}}^{(n)} = {\alpha_{\mathrm{ru}}^{(n)}B}\log_2\left(1+\frac{\left|h_{\mathrm{ru}}^{(n)}\right|^{2}P_{\mathrm{r}}^{(n)}}{\alpha_{\mathrm{ru}}^{(n)}B N_0}\right),
    \end{small}
\end{equation}}%
\endgroup 
where $P_{\mathrm{r}}^{(n)}$ denotes the SemRelay transmit power allocated for bit transmission to the $n$-th user, $\alpha_{\mathrm{ru}}^{(n)}$ denotes the ratio of bandwidth allocated to the SemRelay$\rightarrow$user $n$-th link. 

\subsection{Problem Formulation}
For the considered SemRelay-aided multiuser text transmission system, our target is to maximize the multiuser weighted sum-rate (i.e., effective bit rate) by jointly optimizing the SemRelay transmit power allocation {\small{$\boldsymbol{P_{\mathrm{r}}}\triangleq\{P_{\mathrm{r}}^{(n)},\forall n \in \mathcal{N}\}$}} and system bandwidth allocation {\small{$\boldsymbol{\alpha} \triangleq \{\alpha_{\mathrm{br}},\alpha_{\mathrm{ru}}^{(n)},\forall n \in \mathcal{N}\}$}}. 
This optimization problem can be formulated as
\begingroup\makeatletter\def\f@size{8}\check@mathfonts
{\setlength\abovedisplayskip{1pt}
\setlength\belowdisplayskip{2pt}
\begin{subequations}\label{P1}
        \begin{align}
        {\text{{\fontsize{10pt}{12pt}\selectfont (P1)}}}\max \limits_{\boldsymbol{P_{\mathrm{r}}},\boldsymbol{\alpha},\gamma_{\mathrm{br}}} \;\;
        &\overset{N}{\underset{n=1}{\sum}} \omega^{(n)} {\alpha_{\mathrm{ru}}^{(n)}B}\log_2\left(1+\frac{\left|h_{\mathrm{ru}}^{(n)}\right|^{2}P_{\mathrm{r}}^{(n)}}{\alpha_{\mathrm{ru}}^{(n)}B N_0}\right) \nonumber\\
        \textrm{s.t.}\quad\;\;
        &\overset{N}{\underset{n=1}{\sum}} {\alpha_{\mathrm{ru}}^{(n)}B}\log_2\left(1+\frac{\left|h_{\mathrm{ru}}^{(n)}\right|^{2}P_{\mathrm{r}}^{(n)}}{\alpha_{\mathrm{ru}}^{(n)}B N_0}\right) \nonumber\\
        &\leq \frac{\alpha_{\mathrm{br}}B\mu}{K}\left(a_1+\frac{a_2}{1+e^{-\left(c_1\gamma_{\mathrm{br}}+c_2\right)}}\right), \label{P1(1)} \\
        &\gamma_{\mathrm{br}} = 10\log_{10}\left(\frac{\left|h_{\mathrm{br}}\right|^{2}P_{\mathrm{b}}}{\alpha_{\text{br}}B N_0}\right), \label{P1(2)}\\
        &a_1+\frac{a_2}{1+e^{-\left(c_1\gamma_{\mathrm{br}}+c_2\right)}} \geq \bar{\varepsilon}, \label{P1(3)}\\
        &\alpha_{\mathrm{br}} + \overset{N}{\underset{n=1}{\sum}}\alpha_{\mathrm{ru}}^{(n)} \le 1, \label{P1_a1}\\ 
        &\alpha_{\mathrm{br}} \geq 0, \alpha_{\mathrm{ru}}^{(n)} \geq 0, \forall n \in \mathcal{N}, \label{P1_a2}\\
        & \overset{N}{\underset{n=1}{\sum}} P_{\mathrm{r}}^{(n)} \leq \bar{P}, \label{P1_P1}\\
        &P_{\mathrm{r}}^{(n)} \geq 0, \forall n \in \mathcal{N}, \label{P1_P2}
        \end{align}
\end{subequations}}%
\endgroup
where $\omega^{(n)}$ is the weight for each user $n$, $\bar{P}$ denotes the total transmit power of the SemRelay.
In particular, (\ref{P1(1)}) is the information-causality constraint, i.e., the SemRelay can only forward the information to the users, which has been previously received from the BS in the BS$\rightarrow$SemRelay link. 
Moreover, the constraint (\ref{P1(3)}) enforces the minimum semantic similarity required for the BS$\rightarrow$SemRelay semantic transmission. 
Note that the constraint (\ref{P1(3)}) can also be expressed in the following equivalent form
\begingroup\makeatletter\def\f@size{8}\check@mathfonts
{\setlength\abovedisplayskip{1pt}
\setlength\belowdisplayskip{2pt}
\begin{equation}
    \gamma_{\mathrm{br}} \geq \frac{1}{c_{1}}\ln(\frac{\bar{\varepsilon}-a_1}{a_1+a_2-\bar{\varepsilon}})-\frac{c_2}{c_1} = \bar{\gamma}. \label{P2_(4)}
\end{equation}}%
\endgroup
\noindent The constraints (\ref{P1_a1})--(\ref{P1_P2}) are the total bandwidth constraints and the transmit power constraints for the SemRelay, respectively.

Note that (P1) is a non-convex optimization problem due to the non-convex constraints in (\ref{P1(1)}) and (\ref{P1(2)}). 
Moreover, the optimization variables in problem (P1) are intricately coupled in the constraints (\ref{P1(1)}), thus making problem (P1) challenging to solve. 
To address these issues, we propose an efficient algorithm in the next section to obtain its high-quality suboptimal solution. 

\section{Proposed Algorithm}
To solve problem (P1), we first introduce a set of slack variables $\boldsymbol{u_{\mathrm{ru}}} \triangleq \{u_{\mathrm{ru}}^{(n)},\forall n \in \mathcal{N}\}$. 
As such, we have 
\begingroup\makeatletter\def\f@size{8}\check@mathfonts
{\setlength\abovedisplayskip{1pt}
\setlength\belowdisplayskip{2pt}
\begin{subequations}\label{P2}
    \begin{align}
    {\text{{\fontsize{10pt}{12pt}\selectfont (P2)}}}\;\max \limits_{\boldsymbol{P_{\mathrm{r}}},\boldsymbol{\alpha},\boldsymbol{u_{\mathrm{ru}}}, \gamma_{\mathrm{br}}} &\;\;
    \overset{N}{\underset{n=1}{\sum}} \omega^{(n)} {\alpha_{\mathrm{ru}}^{(n)}B}\log_2\left(1+\frac{\left|h_{\mathrm{ru}}^{(n)}\right|^{2}P_{\mathrm{r}}^{(n)}}{\alpha_{\mathrm{ru}}^{(n)}B N_0}\right) \nonumber\\
    \textrm{s.t.}\quad\;\;\;
    &\overset{N}{\underset{n=1}{\sum}} u_{\mathrm{ru}}^{(n)} \leq \frac{\alpha_{\mathrm{br}}B\mu}{K}\left(a_1+\frac{a_2}{1+e^{-\left(c_1\gamma_{\mathrm{br}}+c_2\right)}}\right), \label{P2_(1)}\\
    &u_{\mathrm{ru}}^{(n)} = {\alpha_{\mathrm{ru}}^{(n)}B}\log_2\left(1+\frac{\left|h_{\mathrm{ru}}^{(n)}\right|^{2}P_{\mathrm{r}}^{(n)}}{\alpha_{\mathrm{ru}}^{(n)}B N_0}\right), \forall n \in \mathcal{N}, \label{P2_0_(2)}\\
    & \text{{\fontsize{10pt}{12pt}\selectfont (\ref{P1(2)}), (\ref{P1_a1})--(\ref{P1_P2}), (\ref{P2_(4)})}}. \nonumber
    \end{align}
\end{subequations}}
\endgroup

The solution to problem (P2) can be obtained by solving a relaxed optimization problem, given as follows.

\vspace{-1mm}
\begin{lemma}\label{lemma_P2}
    \rm{The solution to problem (P2) can be obtained by solving the following problem that relaxes the equality constraints in (\ref{P1(2)}), (\ref{P2_0_(2)}).}
    \begingroup\makeatletter\def\f@size{8}\check@mathfonts
    {\setlength\abovedisplayskip{1pt}
    \setlength\belowdisplayskip{2pt}
    \begin{subequations}\label{P3}
        \begin{align}
        {\text{{\fontsize{10pt}{12pt}\selectfont (P3)}}}\;\max \limits_{\boldsymbol{P_{\mathrm{r}}},\boldsymbol{\alpha},\boldsymbol{u_{\mathrm{ru}}}, \gamma_{\mathrm{br}}} &\;
        \overset{N}{\underset{n=1}{\sum}} \omega^{(n)} {\alpha_{\mathrm{ru}}^{(n)}B}\log_2\left(1+\frac{\left|h_{\mathrm{ru}}^{(n)}\right|^{2}P_{\mathrm{r}}^{(n)}}{\alpha_{\mathrm{ru}}^{(n)}B N_0}\right) \nonumber\\
        \textrm{s.t.}\quad\;\;\;
        &u_{\mathrm{ru}}^{(n)} \geq {\alpha_{\mathrm{ru}}^{(n)}B}\log_2\left(1+\frac{\left|h_{\mathrm{ru}}^{(n)}\right|^{2}P_{\mathrm{r}}^{(n)}}{\alpha_{\mathrm{ru}}^{(n)}B N_0}\right), \forall n \in \mathcal{N}, \label{P2_(2)}\\
        &\gamma_{\mathrm{br}} \leq 10\log_{10}\left(\frac{\left|h_{\mathrm{br}}\right|^{2}P_{\mathrm{b}}}{\alpha_{\text{br}}B N_0}\right), \label{P2_(3)} \\
        & {\text{{\fontsize{10pt}{12pt}\selectfont (\ref{P1_a1})--(\ref{P1_P2}), (\ref{P2_(4)}), (\ref{P2_(1)})}}}. \nonumber
        \end{align}
    \end{subequations}}
    \endgroup

\end{lemma}

\begin{proof} 
If the optimal solution to (P3) satisfies the constraints in (\ref{P2_(2)}) with strict inequality for $\tilde{n}$-th user, we can increase the corresponding SemRelay-user ratio of total bandwidth $\alpha_{\mathrm{ru}}^{(\tilde{n})}$ to reduce the gap between the right-hand-sides (RHS) of the constraint (\ref{P2_(2)}) and $u_{\mathrm{ru}}^{(\tilde{n})}$, followed by decreasing $\alpha_{\mathrm{br}}$ or $u_{\mathrm{ru}}^{(\tilde{n})}$ to re-satisfy the constraints (\ref{P1_a1}) and (\ref{P2_(1)}) if the constraint (\ref{P1_a1}) becomes inactive. 
Note that the objective value of (P3) is non-decreasing, hence there exists an optimal solution to (P3) where all the constraints in (\ref{P2_(2)}) are satisfied with equality. 
Therefore, the optimal value of (P3) is equal to that of (P2).
The constraint (\ref{P2_(3)}) can also be proved similarly.
\end{proof}
\vspace{-2mm}

\addtolength{\topmargin}{0.01 in}
However, problem (P3) is still a non-convex problem due to the constraints in (\ref{P2_(1)}), (\ref{P2_(2)}), and (\ref{P2_(3)}). 
To tackle this issue, we proposed a BCD-based algorithm that splits the optimization variables of the problem into two blocks: 1) SemRelay transmit power allocation $\boldsymbol{P_{\mathrm{r}}}$, and 2) system bandwidth allocation $\boldsymbol{\alpha}$. 
Subsequently, these two blocks are optimized alternately and iteratively with the other one being fixed, until the convergence is achieved. 

\subsection{SemRelay Transmit Power Allocation} \label{power allocation}
For any given bandwidth allocation $\boldsymbol{\alpha}$, problem (P3) reduces to the following problem for optimizing the SemRelay transmit power allocation 
\begingroup\makeatletter\def\f@size{8}\check@mathfonts
{\setlength\abovedisplayskip{1pt}
\setlength\belowdisplayskip{2pt}
\begin{subequations}\label{P4}
    \begin{align}
    {\text{{\fontsize{10pt}{12pt}\selectfont (P4)}}}\max \limits_{\boldsymbol{P_{\mathrm{r}}},\boldsymbol{u_{\mathrm{ru}}}, \gamma_{\mathrm{br}}} \;&\;
    \overset{N}{\underset{n=1}{\sum}} \omega^{(n)} {\alpha_{\mathrm{ru}}^{(n)}B}\log_2\left(1+\frac{\left|h_{\mathrm{ru}}^{(n)}\right|^{2}P_{\mathrm{r}}^{(n)}}{\alpha_{\mathrm{ru}}^{(n)}B N_0}\right) \nonumber\\
    \textrm{s.t.}\quad\;\;
    & \text{{\fontsize{10pt}{12pt}\selectfont (\ref{P1_P1}), (\ref{P1_P2}), (\ref{P2_(4)}), (\ref{P2_(1)}), (\ref{P2_(2)}), (\ref{P2_(3)})}}. \nonumber
    \end{align}
\end{subequations}}
\endgroup

Although problem (P4) is a non-convex optimization problem due to the non-convex constraints in (\ref{P2_(1)}), (\ref{P2_(2)}), we can address them by using the successive convex approximation (SCA) method, as elaborated below.

\vspace{-0.1 in}
\begin{lemma} \label{lower-bound R_br(power)}
    \rm{For the constraint (\ref{P2_(1)})}, $R_{\mathrm{br}}$ is a convex function of $(1+e^{-(c_1\gamma_{\mathrm{br}}+c_2)})$.
    For any local point $\tilde{\gamma}_{\mathrm{br}}$, $R_{\mathrm{br}}$ can be lower-bounded as
    {\setlength\abovedisplayskip{1pt}
    \setlength\belowdisplayskip{2pt}
    \begin{equation} 
        \begin{aligned}
            R_{\mathrm{br}} &\geq E_1-E_2(e^{-(c_1\gamma_{\mathrm{br}}+c_2)}-e^{-(c_1\tilde{\gamma}_{\mathrm{br}}+c_2)}) \\ 
            &\triangleq  R_{\mathrm{br}\mathrm{(lb)}}, \label{lemma1}
        \end{aligned}
    \end{equation}}%
    where the coefficients $E_{1}$ and $E_{2}$ are respectively given by 
    $E_{1} = \frac{\alpha_{\mathrm{br}}B\mu}{K}\left(a_1+\frac{a_2}{1+e^{-\left(c_1\tilde{\gamma}_{\mathrm{br}}+c_2\right)}}\right)$, and 
    $E_{2} = \frac{\alpha_{\mathrm{br}}B\mu}{K}\\
    \left(\frac{a_2}{(1+e^{-(c_1\tilde{\gamma}_{\mathrm{br}}+c_2)})^2}\right)$.
\end{lemma}

\begin{proof}
    Although the RHS of the constraint (\ref{P2_(1)}) is a non-convex function with respect to $\gamma_{\mathrm{br}}$, it is a convex function with respect to $(1+e^{-(c_1\gamma_{\mathrm{br}}+c_2)})$. 
    Since the first-order Taylor expansion of a convex function is its global under-estimator \cite{boyd2004convex}, we can use the SCA method to obtain its lower bound by the first-order Taylor expansion, and thus get the convex constraint with respect to $\gamma_{\mathrm{br}}$. 
\end{proof}

\vspace{-0.2 in}

\begin{lemma} \label{upper-bound R_ru(power)}
    \rm{For the constraint (\ref{P2_(2)})}, $R_{\mathrm{ru}}^{(n)}$ is a concave function of $P_{\mathrm{r}}^{(n)}$. 
    For any local point $\tilde{P}_{\mathrm{r}}^{(n)}$, $R_{\mathrm{ru}}^{(n)}$ can be upper-bounded as
    {\setlength\abovedisplayskip{1pt}
    \setlength\belowdisplayskip{2pt}
    \begin{equation} 
            R_{\mathrm{ru}}^{(n)} \leq E_{3}+E_{4}(P_{\mathrm{r}}^{(n)}-\tilde{P}_{\mathrm{r}}^{(n)}) \triangleq \phi_{\mathrm{(up)}}^{(n)}, \label{lemma2}
    \end{equation}}%
    where the coefficients $E_{3}$ and $E_{4}$ are respectively given by 
    $E_{3} = \alpha_{\mathrm{ru}}^{(n)}B\log_{2}(1+\frac{\left|h_{\mathrm{ru}}^{(n)}\right|^{2}\tilde{P}_{\mathrm{r}}^{(n)}}{\alpha_{\mathrm{ru}}^{(n)}B N_0})$, 
    and 
    $E_{4} = \frac{\alpha_{\mathrm{ru}}^{(n)}B\left|h_{\mathrm{ru}}^{(n)}\right|^{2}}{(\alpha_{\mathrm{ru}}^{(n)}BN_{0}+\left|h_{\mathrm{ru}}^{(n)}\right|^{2} \tilde{P}_{\mathrm{r}}^{(n)})\ln{2}}$.
\end{lemma}

Based on Lemmas \ref{lower-bound R_br(power)} and \ref{upper-bound R_ru(power)}, 
by replacing the RHS of (\ref{P2_(1)}), (\ref{P2_(2)}) with their corresponding lower or upper bounds, i.e., $R_{\mathrm{br}\mathrm{(lb)}}$ in (\ref{lemma1}) and $\phi_{\mathrm{(up)}}^{(n)}$ in (\ref{lemma2}), respectively, problem (P4) can be transformed into the following approximate form
\begingroup\makeatletter\def\f@size{8}\check@mathfonts
{\setlength\abovedisplayskip{1pt}
\setlength\belowdisplayskip{2pt}
\begin{subequations} \label{P5}
    \begin{align}
        {\text{{\fontsize{10pt}{12pt}\selectfont (P5)}}}\max \limits_{\boldsymbol{P_{\mathrm{r}}},\boldsymbol{u_{\mathrm{ru}}}, \gamma_{\mathrm{br}}} &\;\;
        \overset{N}{\underset{n=1}{\sum}} \omega^{(n)} {\alpha_{\mathrm{ru}}^{(n)}B}\log_2\left(1+\frac{\left|h_{\mathrm{ru}}^{(n)}\right|^{2}P_{\mathrm{r}}^{(n)}}{\alpha_{\mathrm{ru}}^{(n)}B N_0}\right) \nonumber\\
        \textrm{s.t.}\quad\;
        &\overset{N}{\underset{n=1}{\sum}} u_{\mathrm{ru}}^{(n)} \leq R_{\mathrm{br}\mathrm{(lb)}}, \\
        &u_{\mathrm{ru}}^{(n)} \geq \phi_{\mathrm{(up)}}^{(n)}, \forall n \in \mathcal{N},\\
        & \text{{\fontsize{10pt}{12pt}\selectfont (\ref{P1_P1}), (\ref{P1_P2}), (\ref{P2_(4)}), (\ref{P2_(3)})}}. \nonumber
    \end{align}
\end{subequations}}
\endgroup

Problem (P5) is now a convex optimization problem, which can be efficiently solved by using CVX solvers \cite{cvx}.

\subsection{Bandwidth Allocation}
For any given SemRelay power $\boldsymbol{P_{\mathrm{r}}}$, problem (P3) reduces to the following problem for optimizing the bandwidth allocation
\begingroup\makeatletter\def\f@size{8}\check@mathfonts
{\setlength\abovedisplayskip{1pt}
\setlength\belowdisplayskip{2pt}
\begin{subequations} 
    \begin{align}
    {\text{{\fontsize{10pt}{12pt}\selectfont (P6)}}} \max \limits_{\boldsymbol{\alpha},\boldsymbol{u_{\text{ru}}}, \gamma_{\text{br}}} \;&\;
    \overset{N}{\underset{n=1}{\sum}} \omega^{(n)} {\alpha_{\mathrm{ru}}^{(n)}B}\log_2\left(1+\frac{\left|h_{\mathrm{ru}}^{(n)}\right|^{2}P_{\mathrm{r}}^{(n)}}{\alpha_{\mathrm{ru}}^{(n)}B N_0}\right) \nonumber\\
    \textrm{s.t.}\quad\; 
    & \text{{\fontsize{10pt}{12pt}\selectfont (\ref{P1_a1}), (\ref{P1_a2}), (\ref{P2_(4)}), (\ref{P2_(1)}), (\ref{P2_(2)}), (\ref{P2_(3)})}}. \nonumber
    \end{align}
\end{subequations}}
\endgroup

Problem (P6) is a non-convex optimization problem due to the non-convex constraints in (\ref{P2_(1)}), (\ref{P2_(2)}), and (\ref{P2_(3)}). 
Specifically, the constraint (\ref{P2_(1)}) exhibits strong coupling between the optimization variables $\alpha_{\mathrm{br}}$ and $\gamma_{\mathrm{br}}$, rendering it challenging to solve this problem. 
To deal with this issue, we first introduce an auxiliary variable 
{\setlength\abovedisplayskip{1pt}
\setlength\belowdisplayskip{2pt}
\begin{equation}
    S = a_1+\frac{a_2}{1+e^{-\left(c_1\gamma_{\text{br}}+c_2\right)}},  \label{eq_S}
\end{equation}}%
\noindent and then equivalently rewrite the constraint (\ref{P2_(1)}) as 
\begingroup\makeatletter\def\f@size{8}\check@mathfonts
{\setlength\abovedisplayskip{1pt}
\setlength\belowdisplayskip{2pt}
\begin{equation} 
    \overset{N}{\underset{n=1}{\sum}} u_{\text{ru}}^{(n)} \le \frac{\alpha_{\text{br}}B\mu S}{K} = \frac{B\mu}{4K}\left((\alpha_{\text{br}}+S)^2-(\alpha_{\text{br}}-S)^2\right). \label{eq_a_S}
\end{equation}}
\endgroup
Although (\ref{eq_a_S}) is still a non-convex constraint, we can address it by using the SCA method, with similar proof to Lemma 2. 

\vspace{-0.1 in}
\begin{lemma} \label{relax R_br(alpha)}
    \rm{For the constraint (\ref{eq_a_S})}, $\delta \triangleq (\alpha_{\mathrm{br}}+S)^2$ is a convex function of $\alpha_{\mathrm{br}}$ and $S$.
    Given any local points $\tilde{\alpha}_{\mathrm{br}}$ and $\tilde{S}$, $\delta$ can be lower-bounded as
    {\setlength\abovedisplayskip{1pt}
    \setlength\belowdisplayskip{2pt}
    \begin{equation}
    	\delta \geq -(\tilde{\alpha}_{\mathrm{br}}+\tilde{S})^2+2(\tilde{\alpha}_{\mathrm{br}}+\tilde{S})(\alpha_{\mathrm{br}}+S) \triangleq \delta_{\mathrm{(lb)}}.
    \end{equation}}
\end{lemma}
\noindent Using Lemma \ref{relax R_br(alpha)}, the constraint (\ref{P2_(1)}) can be transformed as
\begingroup\makeatletter\def\f@size{8}\check@mathfonts
{\setlength\abovedisplayskip{1pt}
\setlength\belowdisplayskip{2pt}
\begin{equation}
    \overset{N}{\underset{n=1}{\sum}} u_{\mathrm{ru}}^{(n)} \le \frac{B\mu}{4K}(\delta_{\mathrm{(lb)}}-(\alpha_{\mathrm{br}}-S)^2). \label{P6_(1)}
\end{equation}}
\endgroup

Next, for the non-affine equality constraint (\ref{eq_S}), we can first relax it as
{\setlength\abovedisplayskip{1pt}
\setlength\belowdisplayskip{2pt}
\begin{equation}
    S \leq a_1+\frac{a_2}{1+e^{-\left(c_1\gamma_{\mathrm{br}}+c_2\right)}}. \label{relax_S}
\end{equation}}%
\noindent One can observe that, given $\boldsymbol{P_{\mathrm{r}}}$, the non-convex constraint (\ref{relax_S}) shares a similar form with the RHS of the constraint (\ref{P2_(1)}).
Based on Lemma \ref{lower-bound R_br(power)}, we can transform the constraint (\ref{relax_S}) as 
{\setlength\abovedisplayskip{1pt}
\setlength\belowdisplayskip{2pt}
\begin{equation} 
        S \leq E_{5}-E_{6}(e^{-(c_1\gamma_{\mathrm{br}}+c_2)}-e^{-(c_1\tilde{\gamma}_{\mathrm{br}}+c_2)}), \label{cons_S}
\end{equation}}%
where the coefficients $E_{5}$ and $E_{6}$ are respectively given by 
$ E_{5} = a_1+\frac{a_2}{1+e^{-\left(c_1\tilde{\gamma}_{\mathrm{br}}+c_2\right)}}$,
and 
$ E_{6} = \frac{a_2}{(1+e^{-(c_1\tilde{\gamma}_{\mathrm{br}}+c_2)})^2}$.

Moreover, for the non-convex constraints (\ref{P2_(2)}) and (\ref{P2_(3)}), they can also be efficiently addressed by using the SCA method.

\begin{lemma} \label{relax R_ru(alpha)}
    \rm{For the constraint (\ref{P2_(2)})}, its RHS $R_{\mathrm{ru}}^{(n)}$ is a concave function with respect to $\alpha_{\mathrm{ru}}^{(n)}$, which can be upper-bounded as
    {\setlength\abovedisplayskip{1pt}
    \setlength\belowdisplayskip{2pt}
    \begin{equation}
            R_{\mathrm{ru}}^{(n)} \leq E_{7}+E_{8}(\alpha_{\mathrm{ru}}^{(n)}-\tilde{\alpha}_{\mathrm{ru}}^{(n)}) \triangleq R_{\mathrm{ru}\mathrm{(up)}}^{(n)}, \label{cons_R_ru}    
    \end{equation}}%
    where the coefficients $E_{7}$ and $E_{8}$ are respectively given by 
    \begingroup\makeatletter\def\f@size{9}\check@mathfonts
    $E_{7} = \tilde{\alpha}_{\mathrm{ru}}^{(n)}B\log_2(1+\frac{\left|h_{\mathrm{ru}}^{(n)}\right|^{2}P_{\mathrm{r}}^{(n)}}{\tilde{\alpha}_{\mathrm{ru}}^{(n)}B N_0})$, 
    \endgroup
    and 
    \begingroup\makeatletter\def\f@size{9}\check@mathfonts
    {$E_{8} = B\log_2(1+\frac{\left|h_{\mathrm{ru}}^{(n)}\right|^{2}P_{\mathrm{r}}^{(n)}}{\tilde{\alpha}_{\mathrm{ru}}^{(n)}B N_0}) - \frac{\log_2(e) B \left|h_{\mathrm{ru}}^{(n)}\right|^{2}P_{\mathrm{r}}^{(n)}}{\tilde{\alpha}_{\mathrm{ru}}^{(n)}B N_0+\left|h_{\mathrm{ru}}^{(n)}\right|^{2}P_{\mathrm{r}}^{(n)}}$}.
    \endgroup
\end{lemma}

\vspace{-0.1 in}
\begin{lemma} \label{relax gamma(alpha)}
    \rm{For the constraint (\ref{P2_(3)})}, $\psi \triangleq 10\log_{10}(\frac{\left|h_{\mathrm{br}}\right|^{2}P_{\mathrm{b}}}{\alpha_{\text{br}}B N_0})$ is a convex function with respect to $\alpha_{\mathrm{br}}$. 
    For any local point $\tilde{\alpha}_{\mathrm{br}}$, $\psi$ can be lower-bounded as
    {\setlength\abovedisplayskip{1pt}
    \setlength\belowdisplayskip{2pt}
    \begin{equation}
            \psi \geq E_{9}-E_{10}(\alpha_{\mathrm{br}}-\tilde{\alpha}_{\mathrm{br}}) \triangleq \psi_{\mathrm{(lb)}}, \label{cons_gamma_alpha}    
        \end{equation}}%
    where the coefficients $E_{9}$ and $E_{10}$ are respectively given by
    $        E_{9} = 10\log_{10}{(\frac{\left|h_{\mathrm{br}}\right|^{2}P_{\mathrm{b}}}{\tilde{\alpha}_{\mathrm{br}}B N_0})}$,        and
   $ E_{10} = \frac{10}{\tilde{\alpha}_{\mathrm{br}}\ln{10}}.   $  
\end{lemma}

Based on Lemmas \ref{relax R_br(alpha)}--\ref{relax gamma(alpha)} and (\ref{P6_(1)}), (\ref{cons_S}), problem (P6) can be transformed into the following approximate problem
\begingroup\makeatletter\def\f@size{8}\check@mathfonts
{\setlength\abovedisplayskip{1pt}
\setlength\belowdisplayskip{2pt}
\begin{subequations} 
    \begin{align}
    {\text{{\fontsize{10pt}{12pt}\selectfont (P7)}}} \max \limits_{\boldsymbol{\alpha},\boldsymbol{u_{\mathrm{ru}}}, \gamma_{\mathrm{br}}, S} &
    \overset{N}{\underset{n=1}{\sum}} \omega^{(n)} {\alpha_{\mathrm{ru}}^{(n)}B}\log_2\left(1+\frac{\left|h_{\mathrm{ru}}^{(n)}\right|^{2}P_{\mathrm{r}}^{(n)}}{\alpha_{\mathrm{ru}}^{(n)}B N_0}\right) \nonumber\\
    \textrm{s.t.}\quad\;\;
    & u_{\mathrm{ru}}^{(n)} \geq R_{\mathrm{ru}\mathrm{(up)}}^{(n)}, \forall n \in \mathcal{N}, \\
    &\gamma_{\mathrm{br}} \leq \psi_{\mathrm{(lb)}}, \label{relax_gamma_a_B}\\
    & \text{{\fontsize{10pt}{12pt}\selectfont (\ref{P1_a1}), (\ref{P1_a2}), (\ref{P2_(4)}), (\ref{P6_(1)}), (\ref{cons_S})}}. \nonumber
    \end{align}
\end{subequations}}
\endgroup

Problem (P7) is now a convex optimization problem, which can be efficiently solved by using CVX solvers.

\subsection{Algorithm Analysis}

We summarize the proposed algorithm for solving problem (P1) in Algorithm 1, where $f(\mathbf{P_r},\boldsymbol{\alpha}) \triangleq f(\boldsymbol{X})$ denotes the objective value of problem (P1) with the solution $\mathbf{P_r}$ and $\boldsymbol{\alpha}$, $\tau$ is the convergence accuracy. 
The optimal solution of (P1) can be obtained by executing steps 3--6 over iterations. 
The convergence and complexity analysis for Algorithm 1 is as follows. 
First, for the solution to problem (P1) $\boldsymbol{X}^{(r)}$ obtained in the $r$-th iteration, the objective value $f(\boldsymbol{X}^{(r)})$ is no smaller than  $f(\boldsymbol{X}^{(r-1)})$ for problem (P1). 
Thus, the objective value of problem (P1) is guaranteed to be non-decreasing over the iterations. 
Moreover, since the objective value has a finite upper bound, the proposed algorithm is guaranteed to converge to a local-optimal solution.
Note that the overall complexity of Algorithm 1 is polynomial in the worst case, since each iteration of Algorithm 1 only needs to solve convex optimization problems. 
Let $N_1$ denote the total number of optimization variables and $I$ denote the number of iterations to achieve convergence. 
Then the overall algorithm complexity can be characterized as $O(IN_{1}^{3.5})$ \cite{boyd2004convex}.
\addtolength{\topmargin}{+1pt}
\begin{algorithm}[!t] \label{algorithm_1}
    \caption{Proposed algorithm for solving problem (P1).}
    \begin{algorithmic}[1]
    \STATE Initialization: Obtain an initial solution $\mathbf{P_r}^{(0)}$, $\boldsymbol{\alpha}^{(0)}$; calculate $R_{\mathrm{sum}}^{(0)} = f(\mathbf{P_r}^{(0)},\boldsymbol{\alpha}^{(0)})$; set $r=0$.
    \REPEAT
    \STATE Update $r = r + 1$.
    \STATE With fixed system bandwidth allocation $\boldsymbol{\alpha}^{(r-1)}$, update the SemRelay transmit power variable $\mathbf{P_r}^{(r)}$ by solving problem (P5).
    \STATE With fixed SemRelay transmit power allocation $\mathbf{P_r}^{(r)}$, update the bandwidth variable $\boldsymbol{\alpha}^{(r)}$ by solving problem (P7). 
    \STATE Calculate $R_{\mathrm{sum}}^{(r)}  = f(\mathbf{P_r}^{(r)},\boldsymbol{\alpha}^{(r)})$ for the $r$th iteration.
    \UNTIL 
    $\frac{R_{\mathrm{sum}}^{(r)}-R_{\mathrm{sum}}^{(r-1)}}{R_{\mathrm{sum}}^{(r-1)}}<\tau$.
    \STATE Obtain $R_{\mathrm{sum}}^{\star} = R_{\mathrm{sum}}^{(r)}$, and the optimal solution is $\mathbf{P_r}^{(r)}$, $\boldsymbol{\alpha}^{(r)}$. 
    \end{algorithmic}
 \end{algorithm}
\section{Numerical Results}

In this section, we evaluate the effectiveness of the proposed algorithm as well as the performance of the proposed SemRelay-aided multiuser text transmission over Rician fading channels. 
The simulation setup is as follows. 
We consider a system with $N = 10$ ground users that are randomly and uniformly distributed within a maximum distance of $40 \ \mathrm{m}$ from the SemRelay. 
The BS-SemRelay distance is $d_{\mathrm{br}} = 60 \ \mathrm{m}$ and the Rician K-factor is 20 dB. 
The reference path loss is $\rho_{0}=-60 \ \mathrm{dB}$ and the path loss exponent is set as $\beta = 3.5$. 
Besides, we assume the same weight for all users, i.e., $\omega^{(n)} = 1, \forall n \in \mathcal{N}$. 
Moreover, we set $K = 5$, and the corresponding parameters of the semantic similarity function in (\ref{sesi}) can be numerically obtained as $a_{1} = 0.3760, a_{2} = 0.5970, c_{1} = 0.2634$, and $c_{2} = -0.8151$ \cite{mu2022semi_NOMA}. 
In order to ensure reliable transmission, the required semantic similarity is set as $\bar{\varepsilon} = 0.9$ \cite{mu2022semi_NOMA}. 
Consider that the ASCII encoding method is utilized to encode each letter, which uses 8-bit binary numbers to signify a character and allows for 256 distinct characters to be represented. 
Assuming that the average number of letters per word is 5, the average number of bits contained in each word is thus $\mu = 40 \ \mathrm{bits/word}$ \cite{yan2022resource}. 
Other parameters are set as $N_{0} = -169 \ \mathrm{dBm/Hz}, \tau = 10^{-6}$, and $P_{\mathrm{b}} = 2\ \mathrm{W}$. 
Furthermore, we consider the following three baseline schemes for performance comparison: 
1) SemRelay with optimized SemRelay transmit power given equal system bandwidth allocation, 
2) SemRelay with optimized system bandwidth allocation given equal transmit power allocation of the SemRelay, 
and 3) conventional DF relay with joint DF relay transmit power and system bandwidth allocation. 

\begin{figure}[!t]
    \centerline
    {\includegraphics[width=0.76\linewidth]{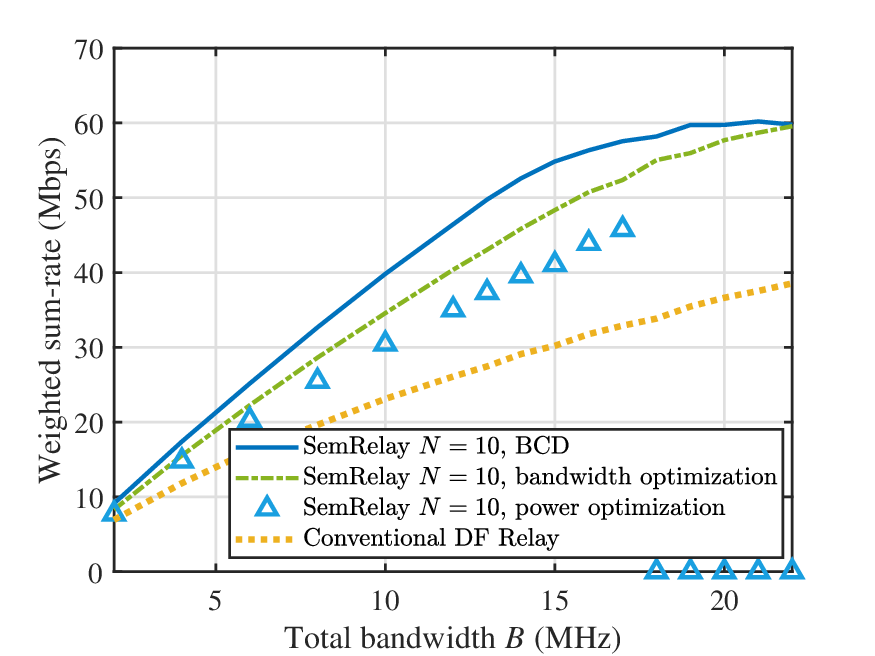}}
    \caption{The achievable weighted sum-rate versus total bandwidth $B$.}
    \label{fig_R_B}
\end{figure}

\begin{figure}[!t]
    \centerline
    {\includegraphics[width=0.76\linewidth]{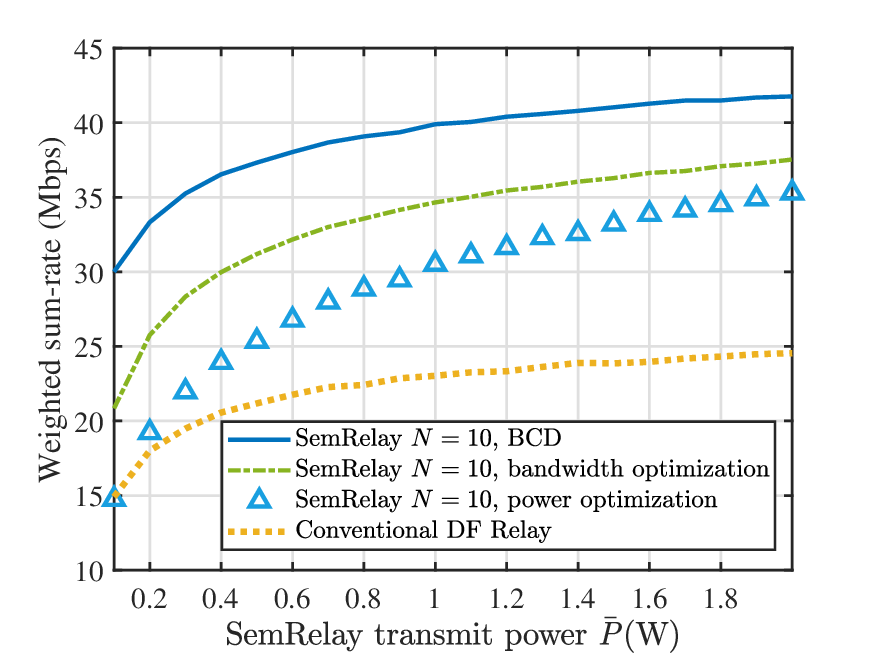}}
    \caption{The achievable weighted sum-rate versus SemRelay transmit power $\bar{P}$.}
    \label{fig_R_P}
    \vspace{-0.15 in}
\end{figure}

Fig. \ref{fig_R_B} presents the achieved (effective) weighted-sum-rate of various schemes versus the total bandwidth $B$, with the total transmit power of SemRelay $\bar{P} = 1 \ \mathrm{W}$. 
Several important observations are made as follows. 
First, the proposed BCD-based algorithm achieves much higher weighted sum-rate than the scheme with power allocation solely, while it has a mile performance gain over the baseline scheme with bandwidth optimization. This indicates that bandwidth allocation has a more dominant impact to improve the weighted sum-rate. 
Next, the SemRelay-aided communication system significantly outperforms the conventional DF relay, especially in the small-bandwidth region. 
This is because the BS-SemRelay transmission link transmits only the semantic information extracted from the original text, leading to higher spectral efficiency. 
It should be noted that the semantic similarity of data transmission is related to the total bandwidth. 
When the total bandwidth increases, the presence of noise in data transmission will have a negative impact on the transmitted semantic symbols, which will lead to a decrease in the semantic similarity. 
Moreover, for SemCom with a stringent requirement, i.e., a high semantic similarity, there is an upper bound of bandwidth to ensure the minimum semantic similarity for high-quality semantic transmission. 
The power optimization scheme with fixed bandwidth cannot achieve the required semantic similarity when the bandwidth is sufficiently large (e.g., $B > 17 \ \mathrm{MHz}$), hence the required data transmission is not possible and the achievable bit rate is 0. 

Fig. \ref{fig_R_P} shows the achieved weighted sum-rate of different schemes versus the SemRelay transmit power $\bar{P}$, with the total bandwidth $B = 10 \ \mathrm{MHz}$. 
It is observed that the weighted sum-rates for all schemes slowly increase with $\bar{P}$ when $\bar{P}$ is larger than $1 \ \mathrm{W}$. 
Moreover, comparing Fig. \ref{fig_R_B} and Fig. \ref{fig_R_P}, it can be concluded that the impact of bandwidth allocation has a greater impact than the SemRelay transmit power allocation in the system weighted sum-rate. 

\section{Conclusions}
In this paper, we proposed a novel SemRelay for improving multiuser text transmission efficiency. 
In particular, the proposed SemRelay equipped with a DeepSC receiver allows semantic transmission over the BS$\rightarrow$SemRelay link and bit transmissions over the SemRelay$\rightarrow$users links. 
For the considered multiuser system, an optimization problem was formulated to maximize the weighted sum-(bit)-rate by jointly designing the SemRelay transmit power allocation and system bandwidth allocation. 
Although this problem is non-convex and hence challenging to solve, we proposed an efficient algorithm to obtain its suboptimal solution by using the BCD and SCA methods. 
Numerical results demonstrated the significant rate performance gain of the proposed SemRelay over conventional DF relay as well as other benchmarks. 
This work can be extended in future works to study the resource management in other SemCom systems such as semantic-aware image/video transmission and knowledge-graph based semantic transmission. 

\vspace{-0.05 in}
\section*{Acknowledgment}
\vspace{-0.05 in}
This work was supported by the National Natural Science Foundation of China under Grant 62201242, 62331023, and Young Elite Scientists Sponsorship Program by CAST 2022QNRC001.
\bibliographystyle{IEEEtran}
\bibliography{reference}

\begin{thebibliography}{10}
\providecommand{\url}[1]{#1}
\csname url@samestyle\endcsname
\providecommand{\newblock}{\relax}
\providecommand{\bibinfo}[2]{#2}
\providecommand{\BIBentrySTDinterwordspacing}{\spaceskip=0pt\relax}
\providecommand{\BIBentryALTinterwordstretchfactor}{4}
\providecommand{\BIBentryALTinterwordspacing}{\spaceskip=\fontdimen2\font plus
\BIBentryALTinterwordstretchfactor\fontdimen3\font minus
  \fontdimen4\font\relax}
\providecommand{\BIBforeignlanguage}[2]{{%
\expandafter\ifx\csname l@#1\endcsname\relax
\typeout{** WARNING: IEEEtran.bst: No hyphenation pattern has been}%
\typeout{** loaded for the language `#1'. Using the pattern for}%
\typeout{** the default language instead.}%
\else
\language=\csname l@#1\endcsname
\fi
#2}}
\providecommand{\BIBdecl}{\relax}
\BIBdecl

\bibitem{yang2022semantic}
W.~Yang, H.~Du, Z.~Q. Liew, W.~Y.~B. Lim, Z.~Xiong, D.~Niyato, X.~Chi, X.~S.
  Shen, and C.~Miao, ``Semantic communications for future internet:
  Fundamentals, applications, and challenges,'' \emph{IEEE Commun. Surv.
  Tutor.}, vol.~25, no.~1, pp. 213--250, first quarter 2023.

\bibitem{SC_magazine_beyond}
G.~Shi, Y.~Xiao, Y.~Li, and X.~Xie, ``From semantic communication to
  semantic-aware networking: Model, architecture, and open problems,''
  \emph{IEEE Commun. Mag.}, vol.~59, no.~8, pp. 44--50, Aug. 2021.

\bibitem{lan2021semantic}
Q.~Lan, D.~Wen, Z.~Zhang, Q.~Zeng, X.~Chen, P.~Popovski, and K.~Huang, ``What
  is semantic communication? {A} view on conveying meaning in the era of
  machine intelligence,'' \emph{J. Commun. Inf. Networks}, vol.~6, no.~4, pp.
  336--371, Dec. 2021.

\bibitem{xie2021deep}
H.~Xie, Z.~Qin, G.~Y. Li, and B.-H. Juang, ``Deep learning enabled semantic
  communication systems,'' \emph{IEEE Trans. Signal Process.}, vol.~69, pp.
  2663--2675, Apr. 2021.

\bibitem{xie2022task}
H.~Xie, Z.~Qin, X.~Tao, and K.~B. Letaief, ``Task-oriented multi-user semantic
  communications,'' \emph{IEEE J. Sel. Areas Commun.}, vol.~40, no.~9, pp.
  2584--2597, Sept. 2022.

\bibitem{zhang2022unified}
G.~Zhang, Q.~Hu, Z.~Qin, Y.~Cai, and G.~Yu, ``A unified multi-task semantic
  communication system with domain adaptation,'' in \emph{Proc. IEEE Global
  Commun. Conf. (GLOBECOM)}, Rio de Janeiro, Brazil, Dec. 2022, pp. 3971--3976.

\bibitem{zhou2022cognitive}
F.~Zhou, Y.~Li, X.~Zhang, Q.~Wu, X.~Lei, and R.~Q. Hu, ``Cognitive semantic
  communication systems driven by knowledge graph,'' in \emph{Proc. IEEE Int.
  Conf. Commun. (ICC)}, Seoul, South Korea, May. 2022, pp. 4860--4865.

\bibitem{xiao2022imitation}
Y.~Xiao, Z.~Sun, G.~Shi, and D.~Niyato, ``Imitation learning-based implicit
  semantic-aware communication networks: Multi-layer representation and
  collaborative reasoning,'' \emph{IEEE J. Sel. Areas Commun.}, vol.~41, no.~3,
  pp. 639--658, Mar. 2023.

\bibitem{wang2022performance}
Y.~Wang, M.~Chen, T.~Luo, W.~Saad, D.~Niyato, H.~V. Poor, and S.~Cui,
  ``Performance optimization for semantic communications: An attention-based
  reinforcement learning approach,'' \emph{IEEE J. Sel. Areas Commun.},
  vol.~40, no.~9, pp. 2598--2613, Sept. 2022.

\bibitem{yan2022resource}
L.~Yan, Z.~Qin, R.~Zhang, Y.~Li, and G.~Y. Li, ``Resource allocation for text
  semantic communications,'' \emph{IEEE Wireless Commun. Lett.}, vol.~11,
  no.~7, pp. 1394--1398, Jul. 2022.

\bibitem{mu2022semi_NOMA}
X.~Mu, Y.~Liu, L.~Guo, and N.~Al-Dhahir, ``Heterogeneous semantic and bit
  communications: A semi-{NOMA} scheme,'' \emph{IEEE J. Sel. Areas Commun.},
  vol.~41, no.~1, pp. 155--169, Jan. 2023.

\bibitem{li2023NOMA_multi_user}
W.~Li, H.~Liang, C.~Dong, X.~Xu, P.~Zhang, and K.~Liu, ``Non-orthogonal
  multiple access enhanced multi-user semantic communication,'' \emph{arXiv
  preprint arXiv:2303.06597}, 2023.

\bibitem{yang2023energy}
Z.~Yang, M.~Chen, Z.~Zhang, and C.~Huang, ``Energy efficient semantic
  communication over wireless networks with rate splitting,'' \emph{arXiv
  preprint arXiv:2301.01987}, 2023.

\bibitem{boyd2004convex}
S.~P. Boyd and L.~Vandenberghe, \emph{Convex optimization}.\hskip 1em plus
  0.5em minus 0.4em\relax Cambridge, U.K.: Cambridge Univ. Press, 2004.

\bibitem{cvx}
M.~Grant and S.~Boyd, ``{CVX}: Matlab software for disciplined convex
  programming, version 2.1,'' \url{http://cvxr.com/cvx}, Mar. 2014.

\end{thebibliography}
\end{document}